%% file: ms.tex
\documentclass[11pt]{article}
    
    \usepackage{pdflscape}
	\usepackage[utf8]{inputenc}
	\usepackage[english]{babel}
	\usepackage{amsfonts}
	\usepackage{amssymb}
	\usepackage{enumitem}
	\usepackage{bbm}
	\usepackage{mathtools}
    \usepackage{comment}
	\usepackage{xcolor}
	\usepackage{subcaption}
	\usepackage{booktabs}	
	\usepackage{setspace}
	\usepackage{graphicx}
	\usepackage{fullpage}
	
	\usepackage{soul}
        \usepackage[backref=page,breaklinks=true,hidelinks]{hyperref}
        \hypersetup{
            colorlinks,
            linkcolor={red!50!black},
            citecolor={blue!50!black},
            urlcolor={blue!80!black}
        }
        \renewcommand*{\backref}[1]{}
        \renewcommand*{\backrefalt}[4]{%
            \ifcase #1 {\footnotesize(Not cited.)}%
            \or        {\footnotesize(Cited on page~#2.)}%
            \else      {\footnotesize(Cited on pages~#2.)}%
            \fi}
            
    \usepackage{natbib}
    \usepackage{amsthm}
    \usepackage{versions}
    \usepackage{thmtools}
    \usepackage{thm-restate}

	\newcommand{\N}{\mathcal{N}}
	\newcommand{\E}{\textnormal{E}}
	\newcommand{\Var}{\textnormal{Var}}
    \renewcommand{\P}{\textnormal{P}}

    \newcommand{\R}{\mathbb{R}}

    \usepackage{siunitx}
    \addto\extrasenglish{%
        
    }
    \makeatletter
    \def\blfootnote{\xdef\@thefnmark{}\@footnotetext}
    \makeatother

    \graphicspath{{Figures/}}

    \title{On the Fairness of Machine-Assisted Human Decisions}
    \author{
        Talia Gillis \\ {\small Columbia University}
        \and
        Bryce McLaughlin \\ {\small Stanford University}
        \and
        Jann Spiess \\ {\small Stanford University}
    }
    \date{
    First version: October 2021 \\
    This version: September 2023}
    
\begin{document}

    \maketitle

    \begin{abstract}
        When machine-learning algorithms are used in high-stakes decisions, we want to ensure that their deployment leads to fair and equitable outcomes. This concern has motivated a fast-growing literature that focuses on diagnosing and addressing disparities in machine predictions. However, many machine predictions are deployed to assist in decisions where a human decision-maker retains the ultimate decision authority. In this article, we therefore consider in a formal model and in a lab experiment how properties of machine predictions affect the resulting human decisions. In our formal model of statistical decision-making, we show that the inclusion of a biased human decision-maker can revert common relationships between the structure of the algorithm and the qualities of resulting decisions. Specifically, we document that excluding information about protected groups from the prediction may fail to reduce, and may even increase, ultimate disparities. In the lab experiment, we demonstrate how predictions informed by gender-specific information can reduce average gender disparities in decisions. While our concrete theoretical results rely on specific assumptions about the data, algorithm, and decision-maker, and the experiment focuses on a particular prediction task, our findings show more broadly that any study of critical properties of complex decision systems, such as the fairness of machine-assisted human decisions, should go beyond focusing on the underlying algorithmic predictions in isolation. 
    \end{abstract}
    
    \blfootnote{
    Talia Gillis (\href{mailto:gillis@law.columbia.edu}{gillis@law.columbia.edu}), Columbia Law School;
    Bryce McLaughlin (\href{mailto:brycem@stanford.edu}{brycem@stanford.edu}) and Jann Spiess (\href{mailto:jspiess@stanford.edu}{jspiess@stanford.edu}), Graduate School of Business, Stanford University.
    We thank Ryan Bubb, Jonah Gelbach, Hoda Heidari, Peter Hull, Daniel Kipnis, Tamar Kricheli-Katz, Sendhil Mullainathan, Mandy Pallais, Asa Palley, Ashesh Rambachan, Stefan Wager, Larry Wein, as well as audiences at the 2021 INFORMS annual meeting, the 2022 FAccT conference, the 2022 MSOM conference, the 2022 Annual Meeting of the American Law and Economics Association (ALEA), the NYU Law \& Economics Workshop, the ETH Zurich and University of Zurich Law \& Economics Workshop, the 2022 Bravo Center Workshop on the Economics of Algorithms at Brown, UW--Madison, Stanford, the 2023 Cowles Conference on Econometrics, the 2023 NBER Summer Institute, and the 2023 Berlin ``Artificial Intelligence and the Economy'' conference for helpful discussions, comments, and suggestions.
    This manuscript supersedes an earlier version that was accepted and presented at FAccT'22, with an abstract published as: Gillis, Talia, Bryce McLaughlin, and Jann Spiess (2022). On the Fairness of Machine-Assisted Human Decisions. In \emph{Proceedings of the 2022 ACM Conference on Fairness, Accountability, and Transparency (FAccT'22)}, page 890.
    }

    \newpage

    \section{Introduction}
    
    When we analyze the properties of machine-learning predictions, we often consider settings in which predictions are automatically translated into decisions.
    In this article, we instead consider how properties of a machine-learning algorithm affect a human decision-maker's choices that take these predictions as an input.
    Focusing on a measure of disparity with respect to the resulting decisions, we document cases in which the interaction of a biased decision-maker with a machine-learning algorithm alters usual trade-offs between fairness and accuracy.
    
    Machine-learning tools are increasingly used in high-stakes decisions.
    Algorithms that predict recidivism are employed in pre-trial bail decisions.
    In medicine, machine-learning predictions are used to make testing decisions.
    In hiring, predictive algorithms screen applicants in order to make interview decisions.
    Lenders use sophisticated statistical models to make default predictions that underlie credit approval decisions.
        
    When machines are used in such high-risk contexts, we often care about their properties beyond accuracy.
    A growing literature specifically studies the fairness properties of machine decisions.
    In medical applications we may take special interest in when an algorithm errs, examining the incidence of false positives and false negatives \citep{mullainathan_diagnosing_2019}.
    In decisions subject to heightened legal oversight such as hiring decisions, we may care about differentiated treatment of or impact on legally designated protected classes
   \citep{raghavan_mitigating_2020}.
    \cite{chouldechova_fair_2016,kleinberg_inherent_2016} study how different fairness criteria apply to machine classifications.
    
    Typically, fairness properties of algorithmic decisions are analyzed as if the machine predictions were implemented directly.
    But in many cases, machine predictions may not be implemented automatically, and instead inform a human decision-maker who has decision-making authority.
    This is the case for recommender systems or decision support systems.
    The human is sometimes considered vital in these systems due to their domain knowledge \citep{lawrence_judgmental_2006}, which encompasses a human decision-maker's intuition and understanding of a problem as well as any additional information they may observe. We may also require that a human make the final decision in a system as a matter of accountability or comfort. In fact, the recent proposal of the European Union for the regulation of artificial intelligence, the most comprehensive and ambitious proposal for the regulation of artificial intelligence to date, requires that humans retain authority when algorithms are used for decision-making in high-risk domains \citep{EU_AI_2021}. 

    In this article, we use a formal model and a lab experiment to ask how questions around fairness and bias play out when machine predictions support rather than replace human decisions, and analyze how incorporating a human decision-maker impacts the relationship between the structure of (machine) predictions and the properties of resulting (human) decisions. In our model, we consider a decision that aims to predict a label from features with minimal prediction loss.
    At the time of the decision, the decision-maker has access to a machine prediction of the label of interest. This prediction is derived from a training data-set that the decision-maker does not have direct access to. This setup could represent a judge's bail decision for which a risk score may be available, or a loan officer's approval decision with the help of a credit score. In our lab experiment, study subjects were incentivized to evaluate the math performance of test-takers based on demographic information as well as an algorithmic assistant that provided summary information from a training sample of similar test-takers.
    
    We focus on the fairness properties of the resulting decision in our model and experiment, and analyze possible trade-offs with accuracy. We capture fairness by the \textit{disparity} that encodes how the decision varies between instances that differ in their membership to a (protected) group. We measure disparity as we may want decisions not to differ between groups, holding other characteristics constant.
    We contrast disparity with \textit{accuracy} (expected loss) of the decision.
    By focusing on conditional statistical parity, we consider a specific notion of fairness that is directly related to the inclusion of protected characteristics in the decision.
    We discuss other notions of fairness in the extension section.
    
    In a baseline case of our theoretical model, where machine predictions are implemented directly, excluding protected characteristics reduces disparities, conditional on all other features. Thus, when the labels differ across groups in the data, the objectives of minimizing risk and minimizing disparity may be in conflict. This presents us with a trade-off to consider when choosing which inputs to include in the algorithm.
    
    Having established the baseline case, we switch to a human decision-maker who uses the machine prediction to inform their decision.
    Adopting a standard model from statistical decision theory, the human decision-maker updates their prior beliefs using the information contained in the machine predictions.
    Thus the role of the predictions moves from taking action to informing the decision-maker through a concise summary (the prediction itself). Changing the inputs available to the prediction now changes the information presented to the decision-maker.
        
    We consider a human decision-maker who starts with biased beliefs about the differences in true labels across protected groups.
    Prior work studies and provides evidence of the role of biased beliefs in discrimination  \citep{coffman_role_2021,bohren_inaccurate_2019}.
    In such cases, the decision-maker believes that differences between groups are larger than they are on average in the training data available to the algorithm.
    The design of the algorithm may then have implications for how the human decision-maker updates their beliefs about the differences between groups before making a decision.

    When the human has a biased prior, common relationships between the inclusion of protected characteristics and the disparities in the final decision may revert.
    We show formally how the decisions of a biased decision-maker may have \textit{larger} disparity when the machine predictions exclude protected characteristics.
    The intuition behind this result is straightforward: if the machine input is not informative about the differences between groups, the prior bias of the human decision-maker remains unmitigated.
    We show this reversal in a stylized example and then provide general theorems that apply in large samples.
    Our results therefore show that, in the case of algorithmic assistance, some of the common trade-offs between fairness and accuracy do not apply.

    We then use an experiment to test how machine-assisted human decisions vary with the type of information presented to decision-makers. In the experiment, study subjects were asked to evaluate the math performance of test-takers based on demographic information (age, gender, education) as well as summary information on the average score of similar test-takers. This average plays the role of an algorithmic assistant in our experiment. Participants were randomized into one of six conditions, which varied by the subsample of test-takers used to calculate the average provided to the participant. The main intervention was variation over whether the assistant was an average calculated for a subsample of men and women (control), or an average calculated separately for men and women (treatment). By randomly varying whether the assistant was gender-blind or gender-aware, we are able to analyze the causal effect of the inclusion of gender information on predictions by our study subjects. The design of the experiment is meant to mirror our model in which the algorithm given to the human decision-maker can be either gender-blind or gender-aware. 

    Preliminary results from our experiment show that the average difference in evaluations for women and men is lower when subjects receive an assistant that is gender-aware. As our main outcome of interest, we consider the difference in the average evaluations for female test-takers and male test-takers in a weighted sample that is meant to reflect the composition of profiles given to experiment participants in the population of test-takers from which these profiles were drawn. For all treatment groups, the disparity in average evaluations for women and men (calculated as the average for women minus the average for men) is negative when the assistant does not include gender, meaning that participants on average predict that men perform better on the math test than women within the weighted sample. However, when the assistant includes gender, the difference in average evaluations between women and men is indistinguishable from zero. 

    We view these results as consistent with our model. When participants receive information that is gender-blind, they estimate that women perform worse on the math test than men, averaged over the test-taker distribution. In reality, women perform better on the math test, so when participants receive an average that is calculated separately for men and women, they adjust their performance estimates for women upwards. This means that disparities of estimates for men and women are lower when information is gender-specific. Based on our preliminary analysis, these results seem to be driven by participants failing to understand gender differences in the relationship between education, age, and math ability, rather than by participants evaluating female test-takers as worse than men with similar profiles. Our results on differences in evaluations across genders thus draw a nuanced picture of the role of including gender in the algorithmic assistant. While the overall results are in line with our model prediction, they are driven by an (implicit) misunderstanding of gender differences in the relationship between education and performance, rather than by outright (explicit) bias against women with otherwise similar profiles.

    Our theoretical and experimental results suggest that we have to take the structure of decisions and beliefs into account when relating machine predictions to the (fairness) properties of algorithms.
    While our analysis shows that typical intuitions about the effect of including protected-class information may not apply to machine-assisted human decisions, there may be other reasons to avoid protected characteristics that we do not capture in our model and that may alter conclusions about their inclusion.
    Importantly, our analysis demonstrates that any such conclusion should take the specific way in which predictions are to be used in the decision-making system into account.
    More broadly, we highlight the importance of analyzing the impact of machine predictions in the context of the full decision environment, which includes the way in which (possibly biased) data informs machine predictions and how humans (with 
    possibly biased priors or preferences) use these predictions.

    We contribute to an interdisciplinary literature on algorithmic fairness that spans computer science, statistics, economics, law, and operations.
    \cite{kleinberg_inherent_2016,chouldechova_fair_2016} study tensions between different fairness qualities of algorithms. 
    \cite{bastani2021improving}, among others, considers the interaction of machine predictions and human decisions.
    \cite{dietvorst_overcoming_2018,stevenson_algorithmic_2019, Ludwig2021-gq} document frictions in the adoption of machine predictions. 
    More specifically, we relate to a literature that explicitly considers fairness properties of machine-assisted human decisions, including \cite{Morgan2019-td,green_disparate_2019}.
    The article is also related to an economics literature on fairness and discrimination, which attempts to identify sources of biased decisions and distinguishes between preference- and belief-based explanations \citep{bordalo_beliefs_2019,bohren_inaccurate_2019,bohren_dynamics_2019,coffman_role_2021,chan_discrimination_2022,eyting_why_2022}.
    In terms of our model, we relate to approaches in econometrics and economics to consider the communication of statistical results \citep{andrews2021model} and the strategic design of information \citep{kamenica_bayesian_2011}.
    We review the overall literature in more detail at the end of this article.
    
    The remainder of this article proceeds as follows. \autoref{sec:setup} introduces our model. \autoref{sec:example} demonstrates our main results in a stylized example before we generalize these results in \autoref{sec:mainresults}.
    Preliminary experimental results are presented in \autoref{sec:experiment}.    
    In \autoref{sec:implications} we discuss some implications of our result.
    Finally, we consider extensions in \autoref{sec:extensions} before expanding on our literature review in \autoref{sec:literature} and concluding in \autoref{sec:conclusion}.
    
    \section{Setup}
    \label{sec:setup}

    This section lays out a model of machine-assisted human decision-making.
    Specifically, we consider a simple prediction decision taken by a decision-maker.
    For an instance $(X,G)$, we consider personalized decisions $\hat{Y} = d(X,G)$ that may vary by features $X$ and group identity $G$. Here, features $X$ may comprise baseline characteristics available at the time the decision is taken, and the group identity $G$ may encode additional sensitive attributes such as gender or ethnicity/race that may be subject to legal constraints or ethical considerations.

    The decision-maker aims to make decisions with small prediction loss.
    A decision $\hat{Y} = d(X,G)$ leads to loss $\ell(\hat{Y},Y)$, where $Y$ is the true label of the instance $(X,G)$ that is unavailable at the time of the decision. For simplicity, we will focus here on a simple prediction decision $\hat{Y} = d(X,G) \in \R$ relative to a true label $Y \in \R$ with squared-error loss $\ell(\hat{Y},Y) = (\hat{Y} - Y)^2$.
    We can think of this choice as an explicit prediction decision or an implicit assessment where loss approximates a consequence of the associated decision. 
    We briefly mention binary decisions in \autoref{sec:extensions}.

    In order to describe the consequences of different decision policies, we make assumptions on the distribution of outcome labels.
    We assume that the true label $Y|X=x,G=g$ has mean $\mu(x,g)$ and variance $\sigma_{x,g}^2$. For tractability, we assume here that the variance is finite, constant, and known ($\sigma_{x,g}^2 \equiv \sigma^2 > 0$), that the error term $Y - \mu(x,g)|X=x,G=g$ is Normally distributed, that $X$ has finite support, and that $G$ is binary.
    The expected (out-of-sample) loss (risk) of the decision rule $d$ is $\bar{r}_d = \E[\ell(d(X,G),Y)] = \E[(d(X,G) - \mu(X,G))^2] + \sigma^2$.
    Writing $r^0_d(x,g) = \E[\ell(d(X,G),Y)|X=x,G=g]$ for the risk at $(x,g)$ and $r_d(x) = \E[\ell(d(X,G),Y)|X=x]$ for the average risk at $x$, we have that $\bar{r}_{d} = \E[r^0_d(X,G)] = \E[r_d(X)]$. The risk expresses the accuracy of the decision.
    
    In addition to the accuracy of the decision, we may also care about how the decision treats instances $(X,G=1), (X,G=0)$ differently that vary only by their group identity.
    We define $\Delta_d(x) = d(x,1) - d(x,0)$ as the decision disparity of decision rule $d$ between group 1 and group 0, and write $\bar{\Delta}_d = \E[\Delta_d(X)]$ for its average across instances.
    
    Disparities $\Delta_d(x)$ may be of interest when there are legal or ethical grounds for not treating people differently who only differ in their group membership. This may be true when the ground-truth discrepancies $\Delta_\mu(x)$ are zero, and we are worried that an unfair decision introduces biases. But even when the true discrepancies are not zero -- which could be due to omitted variables, different outcomes that reflect past or future discrimination, or a causal relationship
    -- we may want to ensure that discrepancies are small to address biased data, correct institutional discrimination, or comply with legal restrictions.
    By considering disparities, we focus on conditional statistical parity as our notion of fairness; we briefly discuss extensions to accuracy-based fairness concepts in \autoref{sec:extensions}.
    
    We assume that the decision $d(x,g)$ is taken by a decision-maker with decision authority who has a belief (prior) $\pi$ over the mean vector $(\mu(x,g))_{x,g}$.
    This belief incorporates any ex-ante beliefs and prior data the decision-maker may have observed.
    In addition to this belief, the decision-maker also observes a machine prediction $\hat{f}(X,G)$, which for simplicity we assume comes from training data of iid draws $(Y_i,X_i,G_i)$ that are independent of the deployment data and independent of the decision-maker's prior (and not available to the decision-maker directly).
    Specifically, we consider the two predictions
    \begin{align*}
        \hat{f}_-(x) &= \frac{1}{n(x,1) + n(x,0)} \sum_{X_i=x} Y_i,
        &
        \hat{f}_+(x,g) &= \frac{1}{n(x,g)} \sum_{X_i=x, G_i=g} Y_i,
    \end{align*}
    where $n(x,g)$ denotes the number of observations $i$ with $X_i=x,G_i=g$.
    The first prediction does not vary with group identity, while the second one does.
    We discuss expanding to more complex functional forms in \autoref{sec:extensions} below, but for now focus on these simple averages, which allow us to model learning by the agent.
    
    We consider three decisions taken by the decision-maker to minimize average expected loss.
    First, we consider the decisions $d_0$ that the decision-maker would take without access to any prediction.
    Second, we consider the decision $\hat{d}_-$ the decision-maker would take when given access to the prediction $\hat{f}_-(X)$ for the instance $(X,G)$.
    Finally, we consider the analogous decision $\hat{d}_+$ given the more detailed prediction $\hat{f}_+(X,G)$.
    We assume that the decision-maker minimizes expected loss $\ell(X,G)$, averaged over her prior.%
    \footnote{We assume that the decision-maker observes a single instance at a time, and at this point takes into account only the prediction for this specific instance. In principle, the decision-maker could also learn from other predictions; however, we assume that the structure of the full prediction function is too complex for or not available to the decision-maker, and the decision-maker solely updates based on $\hat{f}(X,G)$. We discuss extensions in \autoref{sec:extensions}.}
    The three decisions are therefore
    \begin{align*}
        d_0(x,g) &= \E_\pi[\mu(x,g)],
        &
        \hat{d}_-(x,g) &= \E_\pi[\mu(x,g)|\hat{f}_-(x)],
        &
        \hat{d}_+(x,g) &= \E_\pi[\mu(x,g)|\hat{f}_+(x,g)].
    \end{align*}
    Note that $\hat{d}_-(x,g),\hat{d}_+(x,g)$ are random variables since they depend on the training data, even for fixed $x,g$.
    
    We are interested in the relative accuracy and disparities of these three decisions of the decision-maker, and compare them to applying the machine predictions directly.
    While not explicitly discussed in the results below, combining human priors with machine predictions can improve the accuracy of decisions by combining information even when not enforced by institutional constraints.
    In the main part of this article, however, we take human decision authority as given and focus on comparing $d_0,\hat{d}_-,\hat{d}_+$, rather than whether one should or should not delegate a given decision to a human decision-maker, a machine, a human-assisted machine, or a machine-assisted human.
    
    \section{Illustration in a simple example}
    \label{sec:example}

    Having laid out a model of machine-assisted decision-making, we are interested in comparing the consequences of different types of machine assistance.
    We begin by considering a simple example in which we focus on instances with $X=x$ fixed (so that the only variation is the group membership $G$), so for simplicity, we drop the $x$ subscripts and inputs in this section. For this simple example, we assume
    \begin{itemize}
        \item $Y|G = g \sim \N(\mu(g),\sigma^2)$  drawn independently upon realization of $G$, and write $\Delta_\mu = \mu(1) - \mu(0)$ and $\bar{\mu} = \frac{\mu(1) + \mu(0)}{2}$;
        \item the decision-maker holds a prior belief $\pi$ that $\mu(g) \sim \N(\beta(g),\tau^2)$, that these distributions are independent across $g$, and write $\delta = \beta(1) - \beta(0)$ and $\bar{\beta} = \frac{\beta(1) + \beta(0)}{2}$;
        \item the group distribution is $\P(G=1) = \frac{1}{2}$;
        \item we have training data with $n(1)= \frac{n}{2} = n(0)$.
    \end{itemize} 
    In \autoref{sec:mainresults}, we will drop these distributional assumptions about beliefs and generalize the main insights from the example.
    
    Under these assumptions,
    the machine predictions take a particularly simple form, namely
    \begin{align*}
        \hat{f}_+(g) &= \frac{2}{n} \sum_{G_i=g} Y_i,
        &
        \hat{f}_- &= \frac{1}{n} \sum_{i=1}^n Y_i,
    \end{align*}
    and the decision-maker's optimal decisions are
    \begin{align*}
        d_0(g) &= \E_\pi[\mu(g)] = \beta(g)
        = \bar{\beta} + \frac{2g-1}{2} \delta,
        \\
        \hat{d}_-(g) &= \E_\pi[\mu(g)|\hat{f}_-] = \frac{ \sigma^2 \bar{\beta} {+} \frac{\tau^2}{2} \sum_{i=1}^n Y_i}{\sigma^2 + n \frac{\tau^2}{2}} + \frac{2g-1}{2} \delta,
        \\
        \hat{d}_+(g) &= \E_\pi[\mu(g)|\hat{f}_+(g)] = 
        \frac{\sigma^2 \beta(g) {+} \tau^2 \sum_{G_i=g} Y_i}{\sigma^2 + \frac{n}{2} \tau^2}
        .
    \end{align*}
    The decisions $\hat{d}_-$ and $\hat{d}_+$ have an intuitive structure. When given access to $\hat{f}_+(g)$ the decision-maker combines their prior of $\mu(g)$ with a normal signal of $\mu(g)$ whose accuracy is a function of sample size and prior variance. We can find the resulting posterior expectation $\hat{d}_+(g)$ through inverse-variance weighting.
    The decision-maker who only observes $\hat{f}_-$, on the other hand,  updates about the average $\bar{\mu} = \frac{\mu(1) + \mu(0)}{2}$ (first term), but not about the difference $\Delta_\mu = \mu(1) - \mu(0)$ (second term), which is orthogonal in this specific example.

    \newcommand{\greyrule}{    \specialrule{0.25pt}{1pt}{1pt}}
    
    \begin{table}[h]
        \centering
        \begin{tabular}{rcc}
        \toprule
            ${\hat{d}}$ & $\E[\Delta_{\hat{d}}]$ & $\E[r_{\hat{d}}]$ \\
            \midrule
            $\hat{f}_-$ & $0$ & $\frac{\Delta_\mu^2}{4}+\sigma^2 \left(1+\frac{1}{n}\right)$\\
            $\hat{f}_+$ & $\Delta_\mu$ & $\sigma^2 \left(1+\frac{2}{n}\right)$\\
             \midrule
            $d_0$ & $\delta$ & $(\bar{\mu} {-} \bar{\beta})^2 + \frac{(\Delta_\mu {-} \delta)^2}{4} + \sigma^2$ \\
            $\hat{d}_-$ & $\delta$ &
            $
                \frac{\sigma^4 (\bar{\mu} {-} \bar{\beta})^2}{(\sigma^2 + \frac{n}{2} \tau^2)^2} + \frac{(\Delta_\mu {-} \delta)^2}{4} 
                + \sigma^2 \left(1 + \frac{n \tau^4}{4(\sigma^2 + \frac {n}{2} \tau^2)^2}\right)
            $
            \\
            $\hat{d}_+$ & $\frac{\sigma^2 \delta + \frac{n}{2}\tau^2\Delta_\mu}{\sigma^2  + \frac{n}{2}\tau^2 }$ &
            $
                \frac{\sigma^4 \left((\bar{\mu} {-} \bar{\beta})^2 + \frac{(\Delta_\mu {-} \delta)^2}{4} \right)}{\left(\sigma^2 + \frac{n}{2} \tau^2\right)^2} 
                 + \sigma^2 \left(1 + \frac{n \tau^4}{2(\sigma^2 + \frac{n}{2} \tau^2)^2}\right)
            $
                        \\
             \bottomrule
        \end{tabular}
        \caption{Disparities and risks in the example}
        \label{tab:example}
    \end{table}

    For this example, \autoref{tab:example} lists the resulting expected risks and disparities, which are both taken as averages over the training sample given the true values of $\mu(1),\mu(0)$ and at the given $X=x$, in terms of ground truth $\Delta_\mu,\bar{\mu}$, beliefs $\delta,\bar{\beta}$, noise $\sigma^2$, prior variance $\tau^2$, and sample size $n$.
    We first summarize the features of the example in terms of expected disparities; we will generalize these findings and drop expectations in the next section below.
    
    \begin{restatable}[Disparity reversal in the example]{rem}{remexamplereversal}
        \label{rem:examplereversal}
        If the human decision-maker ex-ante believes that the disparity is larger than it actually is, $\delta > \Delta_\mu \geq 0$, then
        \[
            \Delta_{\hat{d}_-}=\Delta_{d_0} = \delta > \E[\Delta_{\hat{d}_+}],
        \]
        while for the underlying predictions
        \(
            \Delta_{\hat{f}_-} = 0 \leq \E[\Delta_{\hat{f}_+}].
        \)
    \end{restatable}
    
    Hence, in this example, the usual ordering of disparities is reversed.
    Under the automated approach that implements $\hat{f}_-,\hat{f}_+$ directly, inclusion of $g$ allows for differentiated action. Thus $\hat{f}_-$ will not exhibit disparity in the actions it takes while $\hat{f}_+$ will on average separate its decisions by $\Delta_\mu$.
    Since the human decision-maker does not update about disparities when given access only to $\hat{f}_-$, the prior disparity persists, leading to excess disparity.
    Giving the biased decision-maker access to $\hat{f}_+$, on the other hand, reduces disparities relative to the unaided decision-maker ($d_0$) as well as the decision-maker who only sees $\hat{f}_-$.
    
    We next inspect the relationship between disparities and accuracy. Fairness and accuracy are sometimes seen as representing a trade-off when considering the inclusion of protected characteristics; we show that this trade-off may disappear under our model of machine-assisted human decisions, and already for pure automation will depend on the bias--variance trade-off in training.
    
    \begin{restatable}[Trade-off reversal in the example]{rem}{remexampletradeoff}
        \label{rem:exampletradeoff}
        For $\xi = \frac{2\sigma}{\sqrt{n}}$ and
        $\delta > \Delta_\mu + \frac{2\tau\sigma}{\sqrt{n\tau^2 + 4\sigma^2}}$, we find that:
    \begin{enumerate}

            \item \ul{Trade-off regime.}    If $\Delta_\mu > \xi$ then there is a disparity--accuracy trade-off for automation,
            \begin{align*}
                |\Delta_{\hat{f}_+}(x)| &> |\Delta_{\hat{f}_-}(x)| \text{ a.s.},
                &
                \E[r_{\hat{f}_+}] &< \E[r_{\hat{f}_-}],
            \end{align*}
            which disappears for assistance,
            \begin{align*}
                \Delta_{\hat{d}_+}(x) &< \E[\Delta_{\hat{d}_-}(x)],
                &
                \E[r_{\hat{d}_+}] &< \E[r_{\hat{d}_-}].
            \end{align*}
            
            \item \ul{Dominance regime.} If $0 \leq \Delta_\mu < \xi$ then there is no trade-off for either decision, and
            \begin{align*}
                |\Delta_{\hat{f}_+}(x)| &> |\Delta_{\hat{f}_-}(x)| \text{ a.s.},
                &\E[r_{\hat{f}_+}] & > \E[r_{\hat{f}_-}],
                \\
                \Delta_{\hat{d}_+}(x) &< \E[\Delta_{\hat{d}_-}(x)],
                &
                \E[r_{\hat{d}_+}] &< \E[r_{\hat{d}_-}].
            \end{align*}
        \end{enumerate}
        
    \end{restatable}
    
    When the sample size is large enough (and thus $\xi$ is small) relative to the disparity in the data, including the group identity in prediction makes the automated decision more accurate at the cost of disparity, while the exclusion of information makes both disparity and accuracy worse for assistance.
    If true disparities are very small, then even for automation there is no trade-off between these two goals.
    By excluding group information, the automated rule reduces variance more than it increases loss due to (statistical) bias.
    In cases where the prior disparity is very close to the true disparity in the data in the sense that $\Delta_\mu < \delta < \Delta_\mu + \frac{2\tau\sigma}{\sqrt{n\tau^2 + 4\sigma^2}}$,
    the assisted human decision even exhibits the reverse of the usual trade-off, where the exclusion of group information increases disparity and accuracy.

    \section{Main results}\label{sec:mainresults}
        
    In this section, we analyze general patterns in the disparities (and accuracy) in the model from \autoref{sec:setup}.
    As a baseline, we first consider the direct implementation of the predictions $\hat{f}_+,\hat{f}_-$ that vary by their use of group=specific information.
    For our main results, we then consider decisions $\hat{d}_+,\hat{d}_-$ by the human decision-maker using these different predictions as input.
    Throughout, we hold the $X$ covariate fixed, $X=x$, as we evaluate the decision-rule, which is without loss in our framework.
    
    If we directly implemented the machine predictions $\hat{f}_+(x,g)$ (which varies by group) and $\hat{f}_-(x)$ (which ignores group membership), we note that disparities would trivially be
    \begin{align*}
        \Delta_{\hat{f}_-}(x) &= 0, &
        \E[\Delta_{\hat{f}_+}(x)] &= \Delta_\mu(x)
    \end{align*}
    (with expectations over the training data), and
    \begin{align}
    \label{eqn:machinedisparity}
        |\Delta_{\hat{f}_+}(x)| > |\Delta_{\hat{f}_-}(x)| = 0
    \end{align}
    almost surely, since $\hat{f}_+$ can vary by group, while $\hat{f}_-$ does not.
    Often, these differences in discrepancy are seen as one side of a trade-off with accuracy, where the more disparate rule is also more accurate; we note that if we take the perspective that $\hat{f}_-,\hat{f}_+$ are learned from training data, then there is not necessarily a trade-off:
    
    \begin{restatable}[Trade-off vs dominance regimes for machine decisions]{rem}{remtradeoffvsdominance}
    \label{rem:tradeoffvsdominance}
        For every $x$ and sample sizes $n(x,1),n(x,0) > 0$
        there exists an $\xi > 0$ such that:
        \begin{enumerate}
            \item \ul{Trade-off regime.} If $\left|\Delta_\mu(x)\right| > \xi$, then
            \begin{align*}
                \left|\Delta_{\hat{f}_+}(x)\right| &> |\Delta_{\hat{f}_-}(x)| \text{ a.s.},
                &
                \E[r_{\hat{f}_+}(x)] & < \E[r_{\hat{f}_-}(x)].
            \end{align*}
            \item \ul{Dominance regime.} If $ \left|\Delta_\mu(x)\right| < \xi$, then
            \begin{align*}
                \left|\Delta_{\hat{f}_+}(x)\right| &> |\Delta_{\hat{f}_-}(x)| \text{ a.s.},
                &
                \E[r_{\hat{f}_+}(x)] & > \E[r_{\hat{f}_-}(x)].
            \end{align*}
        \end{enumerate}
    \end{restatable}

    In other words, when the true disparity in outcomes between the two groups is small, then there is no trade-off between increasing accuracy and increasing disparities; in that case, ignoring group identity and learning jointly serves as a form of regularization that improves predictions in this simple learning framework.
    
    We now turn to human decisions with machine assistance, where we investigate the interplay between decision-maker biases and the nature of machine assistance.
    Specifically, we assume that the decision-maker, while aiming to maximize accuracy, may have biased beliefs about the relative means of the two groups, which we express by (excess) disparity in their prior.
    To formalize our assumption, it will be helpful to define as
    \begin{align*}
        \bar{\mu}(x) = \frac{n(x,1) \mu(x,1) + n(x,0) \mu(x,0)}{n(x,1) + n(x,0)}
    \end{align*}
    the average of instances with $X=x$ in the training data.
    
    \begin{restatable}[$\delta$-disparate beliefs]{ass}{assdeltadis}
        The decision-maker's belief $\pi$ about means at $X=x$ assumes that there is a disparity of at least $\delta > 0$ between groups $G=1$ and $G=0$ with all else known,
        \begin{align*}
            \E_\pi[\mu(x,1) - \mu(x,0) | \bar{\mu}(x)] \geq \delta.
        \end{align*}
    \end{restatable}
    
    In this assumption, we condition on the average $\bar{\mu}(x)$ to rule out cases where beliefs about the difference $\mu(x,1) - \mu(x,0)$ in means are overwhelmed by updates about the average $\bar{\mu}(x)$.
    We also make regularity assumptions about the prior, in addition to Normal error terms in \autoref{sec:setup}.
    
    \begin{restatable}[Finite prior moments]{ass}{assfinitesec}
        For all $(x,g)$, $\E_\pi[|\mu(x,g)|] < \infty$.
    \end{restatable}
    
    Under these assumptions, we $\pi$-almost surely obtain a reversal in disparities of $\hat{d}_-$ and $\hat{d}_+$, relative to $\hat{f}_-$ and $\hat{f}_+$:
    
    \begin{restatable}[Disparity reversal]{thm}{thmdisreversal}
    \label{thm:disreversal}
        Assume that the decision-maker has $\delta$-disparate beliefs, that the regularity conditions hold, and that $0\leq\Delta_\mu(x) < \delta$.
        Then $\pi$-almost surely for every $\eta > 0$ there exists some $N$
        such that with probability (over draws of the training data) at least $1 - \eta$ we have 
        \begin{align*}
            \Delta_{\hat{d}_+}(x) < \delta \leq \Delta_{\hat{d}_-}(x), \Delta_{d_0}(x)
        \end{align*}
         whenever  $n(x,1),n(x,0) \geq N$.
    \end{restatable}
    
    The statement holds for a set of true means $\mu$ that have prior probability one, which may exclude areas of the parameter space that the human decision-maker rules out ex-ante. Specifically, the conclusion of the theorem may fail to hold when the decision-maker has a dogmatic prior that cannot be overcome by the data.
    
    The intuition behind this result is straightforward: If the decision-maker is biased in the sense that they overestimate the disparity relative to the data, then a prediction that does not vary by group preserves that disparity, while separate predictions help overcome it. Relative to the machine baseline in \eqref{eqn:machinedisparity}, the disparity is reversed; including protected characteristics increases disparities in automation, while reducing bias relative to the no-assistance and $\hat{f}_-$-assisted decisions of the human decision-maker.
    
    \begin{restatable}[Disparity reordering]{cor}{cordispreord}
    \label{cor:dispreord}
        Under the conditions of \autoref{thm:disreversal} with $\Delta_\mu(x) \geq 0$, with probability at least $1-\eta$
        \begin{align*}
        \left.
                \begin{array}{c}
                |\Delta_{\hat{d}_-}(x)| \\
              |\Delta_{\hat{d}_0}(x)| \\
            \end{array} \right\}
          >
          \left\{
         \begin{array}{c}
         |\Delta_{\hat{d}_+}(x)| \\
              |\Delta_{\hat{f}_+}(x)|
         \end{array}
         \right\}
            > |\Delta_{\hat{f}_-}(x)|.
        \end{align*}
    \end{restatable}

    When there is a trade-off between disparity and accuracy in the case where the machine prediction is directly implemented,
    then this trade-off is eliminated in large samples as we move from automation to assistance.
    
    \begin{restatable}[Trade-off reversal]{thm}{thmtradeoffrev}
        Assume that the decision-maker has $\delta$-disparate beliefs, that the regularity conditions hold, and that $0 < \Delta_\mu(x) < \delta$.
        Then $\pi$-almost surely for every $\eta > 0$ and $\zeta \in \left(0,\frac{1}{2}\right]$ there exists some $M$
        such that with probability (over draws of the training data) at least $1 - \eta$ we have that 
        \begin{align*}
            \Delta_{\hat{d}_+}(x) &< \Delta_{\hat{d}_-}(x)
            &
            &\text{ and}
            &
            r_{\hat{d}_+}(x) &< r_{\hat{d}_-}(x)
        \end{align*}
        while
        \begin{align*}
            \Delta_{\hat{f}_+}(x) &> \Delta_{\hat{f}_-}(x)
            &
            &\text{ and}
            &
            r^0_{\hat{f}_+}(x,g) &< r^0_{\hat{f}_-}(x,g) \text{ for } g \in \{0,1\}
        \end{align*}
        whenever $\zeta \leq \frac{n(x,0)}{n(x,0) + n(x,1)}, \frac{n(x,1)}{n(x,0) + n(x,1)} \leq 1 - \zeta$ and $n(x,0) + n(x,1) \geq M.$
    \label{thm:trade_reversal}
    \end{restatable}

    Hence, not only does incorporating the biased human decision-maker reverse the intuition around the effect of the inclusion of group information in prediction functions.
    Taking into account the interaction of prediction and human bias also negates the usual trade-off between accuracy and disparity.

    \section{Experiment (preliminary and incomplete)}
    \label{sec:experiment}

    Above, we analyzed machine-assisted human decisions in a simple model of statistical decision-making.
    We now describe a lab experiment that evaluates empirically how machine-assisted human decisions vary with the type of information presented to decision-makers.
    In the experiment, study subjects were asked to evaluate the math performance of test-takers based on demographic information as well as an algorithmic assistant that provided summary information from a training sample of similar test-takers.
    By randomly varying whether the summary information was gender-blind or gender-aware, we are able to analyze the causal effect of the inclusion of gender information on evaluations by our study subjects.
    Preliminary results of our experiment show that evaluations informed by gender-specific information can reduce average gender disparities in human decisions relative to human decisions informed by a gender-blind algorithm.

    \subsection{Experimental design}

    In our experiment, we asked participants to predict the performance of multiple test-takers on a math task.%
    \footnote{The preregistration of the experiment can be found at \href{https://www.socialscienceregistry.org/trials/10416}{socialscienceregistry.org/trials/10416}. The experiment was approved by the IRB at Columbia (approval number AAAU3601) and Stanford (approval number 66199).}
    For each test-taker, the study participants saw (i) test-taker characteristics (age, gender, education) as well as (ii) assistance in the form of an average of the performance of other test-takers with similar characteristics. The main intervention was variation over which group the performance of previous test-takers was averaged to create the assistant that study participants received. In the treatment conditions, averages were formed separately by gender (we consider the self-identified gender of the test-takers as the `protected characteristic'). In the control conditions, averages were taken jointly across genders.
    Participants looked at a series of up to twelve randomly ordered profiles of math test takers from a previous study and evaluated their performance on the test. Participants were rewarded based on the accuracy of their evaluations. Participants were recruited from the Prolific crowdsourcing platform and a total of $1,241$ participants took part in the experiment, making $14,248$ guesses. 
    
    We first describe the population of test-takers used in the experiment and then describe the treatments of the main experiment and outcomes of interest.

    \paragraph{Population of test-takers.}
    The test-takers were taken from a separate experiment in which participants were asked six math questions in addition to other demographic information such as the test-taker's gender, age, education, and income.%
    \footnote{Test scores and profile data are from Cecilia Ridgeway of Stanford University and Tamar Kricheli-Katz of Tel Aviv University for their working paper ``Behavioral responses to the changing world of gender.''}
    The original study had a total of 396 participants, with 207 male participants and 189 female participants. \autoref{fig:sumstats} shows the math score (share of correct answers) of math-takers separately by gender, age (below and above $45$), and education (4-year college degree or less). 

  \begin{figure}[h]
        \centering
        \includegraphics[width=0.7\textwidth]{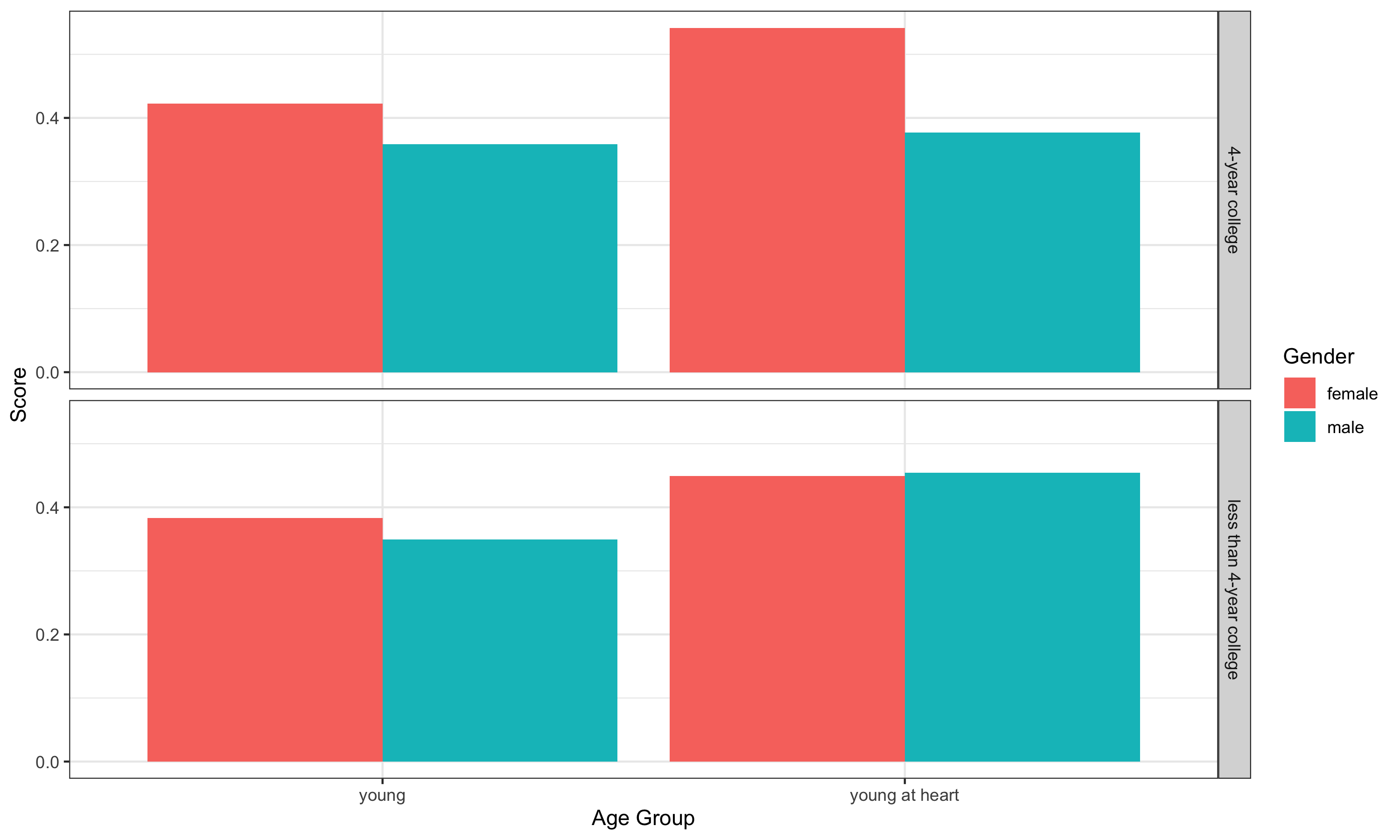}
        \caption{Average math scores (fraction of correct answers, where 1 means all questions were answered correctly) in a data-set of 396 test-takers from a previous study, separated by age (above and below 45), gender (female and male), and education (4-year college degree or less).}
        \label{fig:sumstats}
    \end{figure}
    
    \paragraph{Experimental treatments.}
    All study participants received baseline information about the composition of test-takers and overall average performance. 
    For each test-taker, study participants were shown the test-taker’s gender $g$, and two additional covariates, $x_1$ (age) and $x_2$ (educational attainment). Participants were randomized into one of six conditions, which varied by the subsample of test-takers used to calculate the average provided to the participant, which we refer to as the assistant: 

    \begin{enumerate}
        \item \underline{Condition $0-$}: Participants received average score of all training test-takers. 
        \item \underline{Condition $0+$}: Participants received the average score of test-takers in the training sample who share $g$ with the profile they are evaluating. 
        \item \underline{Condition $1-$}: Participants received the average score of test-takers in the training sample who share the same age range ($x_1$) with the profile they are evaluating. 
        \item \underline{Condition $1+$}: Participants received the average score of test-takers in the training sample who share the same age range ($x_1$) and $g$ with the profile they are evaluating. 
        \item \underline{Condition $2-$}: Participants received the average score of test-takers in the training sample who share the same age range ($x_1$) and educational attainment ($x_2$) with the profile they are evaluating. 
        \item \underline{Condition $2+$}: Participants received the average score of test-takers in the training sample who share the same age range ($x_1$) and educational attainment ($x_2$) and $g$ with the profile they are evaluating. 
    \end{enumerate}
    
    These conditions, therefore, vary by which covariates assistant averages were taken over, $x_1$ (age) and/or $x_2$ (educational attainment), and by whether the assistant is calculated without gender (control) or with gender (treatment). Participants also answered questions about their beliefs in their own prediction abilities and the extent to which they adjusted their predictions in response to the assistant provided. Finally, participants also answered questions about their own gender, age, and education level to allow us to check for in-group biases.  
    
    \paragraph{Outcomes of interest.}
    Our primary outcome of interest is the average evaluations of female test-takers and the average evaluations of male test-takers across the six conditions. As our main test, we compare the difference between the average evaluation of female test-takers and the average evaluation of male test-takers across all treatment conditions (conditions $0+$, $1+$, and $2+$) to the same difference across all control conditions (conditions $0-$, $1-$, and $2-$). Our null hypothesis is that the difference between the average evaluations of female test-takers and male test-takers is not higher in treatment conditions than in control conditions.

\subsection{Preliminary results}

To analyze the differences in evaluations of female test-takers and male test-takers, we consider weighted averages over the evaluation of the 24 different test-taker profiles in our study. The 24 profiles shown to participants consist of a panel of 12 female and 12 male profiles that are balanced with respect to age and education. Because women are less likely to have a 4-year college degree than men in our dataset of test-takers, we oversampled educated women in our profiles selected for the experiment. In analyzing the results, we undo this oversampling by weighting profiles by the share of that group (defined by ten-year age groups, education, and gender) in the population of test-takers. 
We then compare these results to unweighted averages in order to understand the impact of variation in test-taker composition by gender.

Our main results are represented in \autoref{fig:weightedresults}, which shows the average difference in evaluation for women and men (calculated as the average for women minus the average for men, so that a negative difference means that the average was higher for men than women) across different conditions.
For all treatment groups, the difference in the average evaluation for women and men (blue) is negative when the assistant does not include gender, meaning that participants on average predicted that men performed better on the math test than women within the weighted sample.
This negative bias gets attenuated as additional conditioning covariates are added to the assistant averages.
However, when the assistant includes gender, the difference in average evaluations between women and men is indistinguishable from zero across all information sets.
This result stands in contrast to the properties of the algorithmic predictions themselves (red), which are balanced across genders when gender is not included, and show higher averages for women than men when gender is included. This is because women perform better than men on average over the full sample (green).

  \begin{figure}[h]
        \centering
      \includegraphics[width=\textwidth]{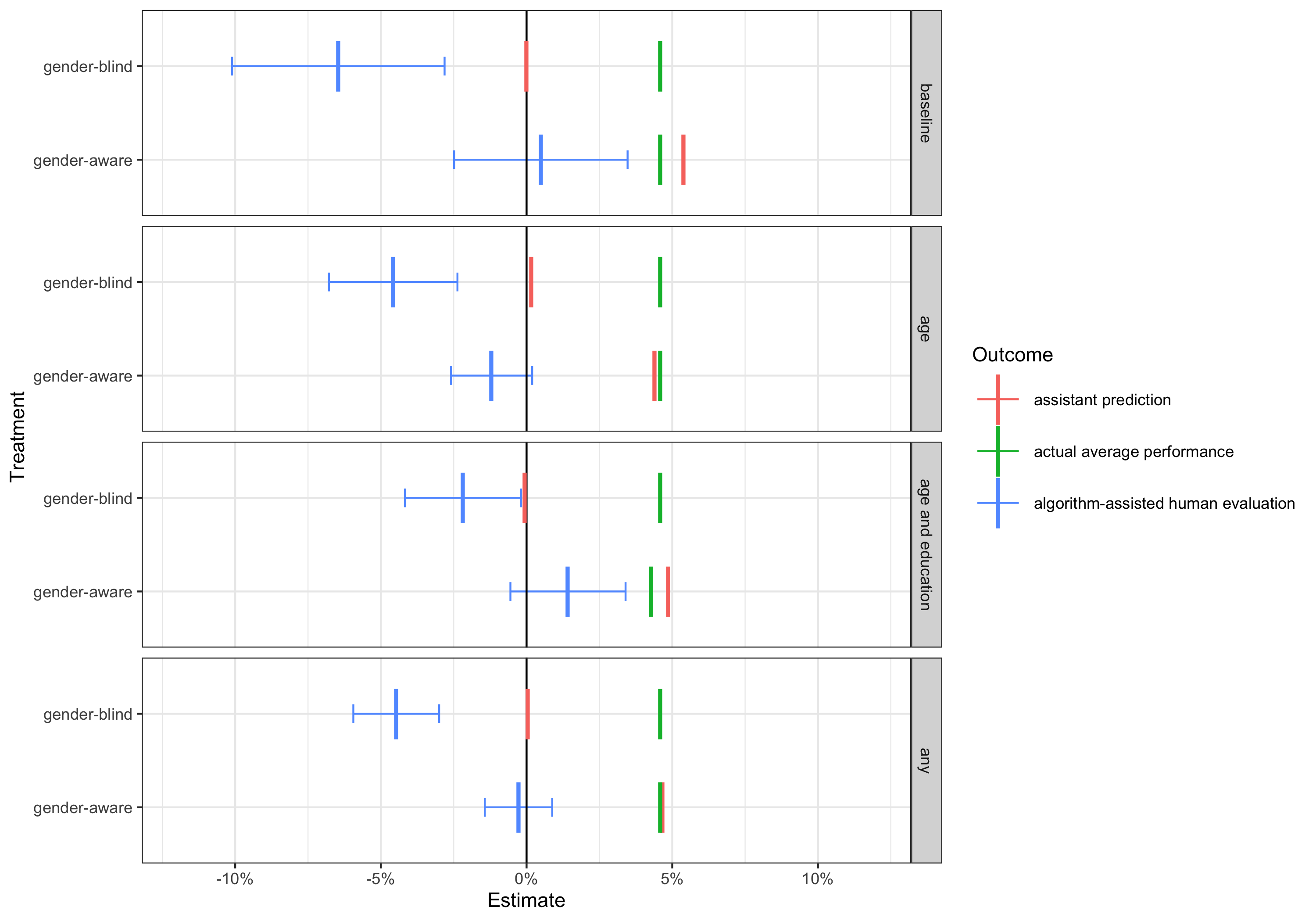}
        \caption{Average evaluation (blue), math score (green), and assistant score (red) based on weighted averages over test-takers in the held-out samples, where weights are chosen according to the share of the respective group in the population of test-takers. The reported estimates represent the percentage-point difference in averages over women minus averages over men, meaning that a positive number indicates that women are evaluated (or perform) relatively higher. Percentages indicate the fraction of correctly answered questions on the math test.}
        \label{fig:weightedresults}
    \end{figure}

We view these results as consistent with our model. When participants received information that was gender-blind, they estimated that women performed worse on the math test than men, averaged over the test-taker distribution. In reality, women performed better on the math test, so when participants received an average that was calculated separately for men and women, they adjusted their performance estimates for women upwards relatively more. This means that disparities of estimates for men and women are lower when information is gender-specific.

 \begin{figure}[h]
        \centering
      \includegraphics[width=\textwidth]{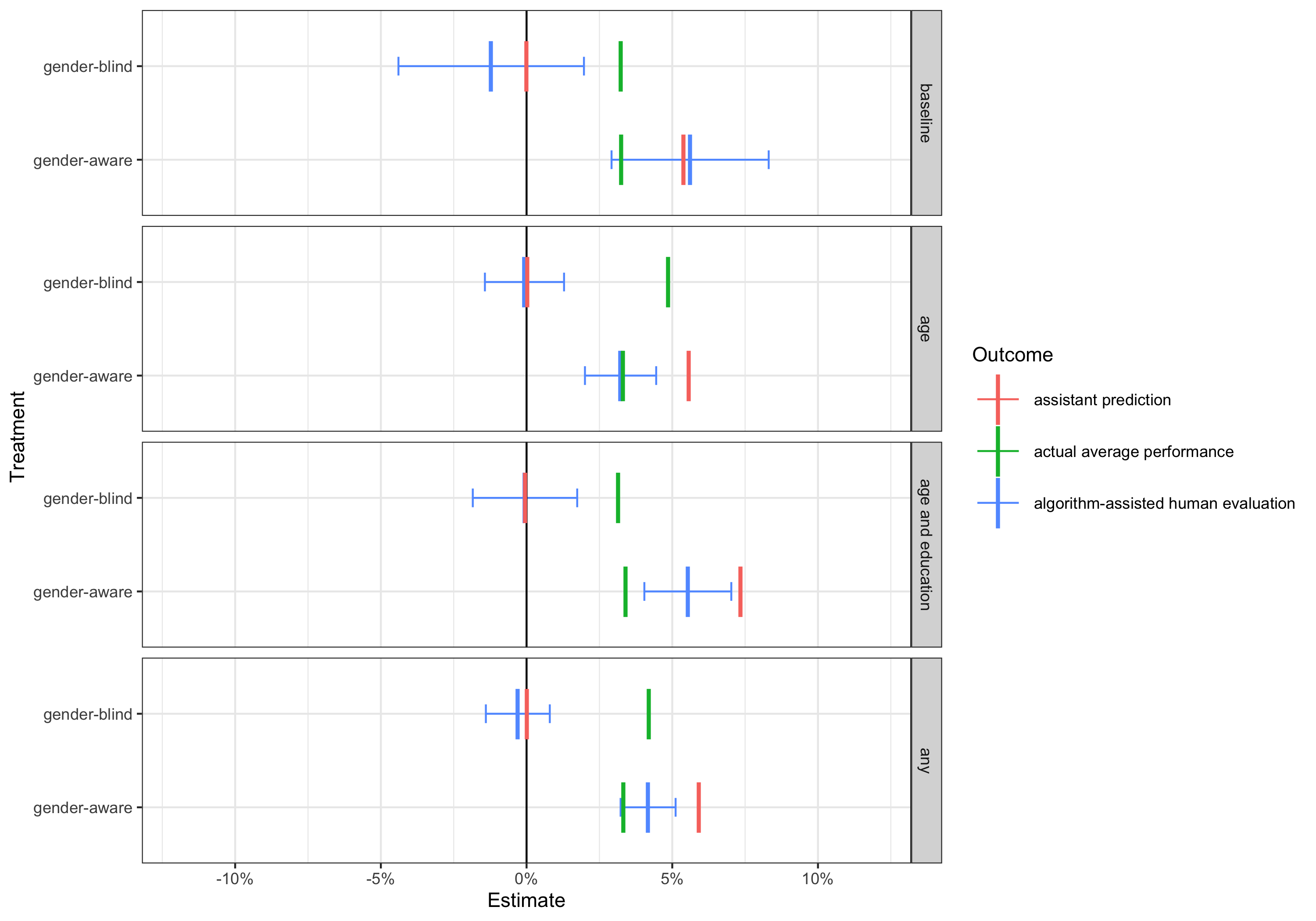}
        \caption{Average evaluation (blue), math score (green), and assistant score (red) based on unweighted average test-taker evaluations of participants, for which profile compositions are artificially balanced across genders. The reported estimates represent the percentage-point difference in averages over women minus averages over men, meaning that a positive number indicates that women are evaluated (or perform) relatively higher. Percentages indicate the fraction of correctly answered questions on the math test.}
        \label{fig:unweightedresults}
    \end{figure}

However, based on our preliminary analysis, these results seem to be driven by participants failing to understand gender differences in the relationship between education, age, and math ability, rather than by participants evaluating female test-takers as worse than men with similar profiles.  \autoref{fig:unweightedresults} shows results for the comparison of raw averages, where profiles are balanced across genders with respect to age and education.
In this balanced sample of test-takers, the study subjects in our study evaluated men and women similarly on average when the assistant was gender-blind, and adjusted their evaluation of women relatively more upward when given gender-specific information. Relative to the true average performance, providing gender-specific assistance, therefore, reduces human bias against women. But in absolute terms, providing gender information increases the average difference between evaluations of women and men with otherwise similar age and education.

Taken together, the weighted and unweighted results on differences in evaluations across genders draw a nuanced picture of the role of including gender in the algorithmic assistant.
While the weighted results are in line with our model prediction, they are driven by an (implicit) misunderstanding of gender differences in the relationship between education and performance, rather than by outright (explicit) bias against women with otherwise similar profiles.
Capturing the results more broadly would require further refinement of the theory that moves beyond conditional parity, and considers implications for unconditional parity when the distribution of characteristics is not balanced across groups.

\section{Implications}
    \label{sec:implications}
    
    In the previous sections, we showed that considering a biased decision-maker who is assisted by a prediction algorithm may reverse the effect of including protected characteristics in the prediction:
    While an algorithm that does not differentiate by protected characteristics may reduce disparities when applied directly (``automation''), excluding group information may not be effective or even counterproductive for reducing disparities when the algorithm provides an input to human decision-making (``assistance'').
    In this section, we consider implications of this result for the debate around the use of protected characteristics in algorithms and the evaluation of algorithmic properties in context.

\subsection{Use of protected characteristics}

    Excluding characteristics from consideration by a machine-learning algorithm -- only in deployment or in training and deployment  -- is sometimes considered a way to assure that people with otherwise similar features are treated similarly.%
    \footnote{A growing literature questions the exclusion of group information because it may be ineffective \citep{gillis_big_2018,gillis_input_2021} or even counterproductive \citep{kleinberg2018algorithmic}. We focus here on the direct effects of exclusion on differences in decisions even in the simple case where exclusion would be effective for automation.}
    Our results show that this is not necessarily true for machine-assisted human decisions and thus provides a potential rationale for including protected characteristics. However, there are other economic, legal, and ethical considerations around the inclusion of protected information that we do not model and may lead to different conclusions, including bias in existing data, biased preferences (as opposed to beliefs) of the agent, dogmatic beliefs that are hard to overcome with data, and limited rationality that prevents the agent from updating towards less biased choices. We review some of these extensions in \autoref{sec:extensions} below.
    
    While our analysis remains limited to the specific case of a human with a biased prior, the formal results clarify that studying the properties of machine-assisted human decisions requires careful consideration of the role of the human decision-maker and the context in which the decision takes place.
    In particular, we document that in a standard model of choice under uncertainty, properties of the prediction function do not directly translate into analogous properties of resulting decisions \citep[similar to the work of][for reporting statistical results to a decision-maker]{andrews2021model}, and that a naive application of existing results may have unintended consequences when applied to machine-assisted human decisions.

    \subsection{Fairness in context}
    
    Our results clarify that determining the fairness, bias, and impact of algorithmic decisions requires going beyond the properties of a prediction rule in isolation.
    Specifically, in the chain from data to prediction to decision (\autoref{fig:context}), biases in the data and properties of prediction rules are not the only factors that shape the properties of the resulting decisions. The process by which data is transformed into a prediction function (the algorithm) and how predictions are used to make decisions (in our case, the human decision-maker with their beliefs) shape the properties of the decision.
    For example, \autoref{rem:tradeoffvsdominance} shows that the training process (in this case, the signal-to-noise ratio in the data and implicit regularization) matters for the relationship between accuracy and disparity. 
    Our main result then shows that the mapping from prediction function to human decisions is non-trivial and may likewise alter commonly assumed trade-offs.
    
        \begin{figure}[h]
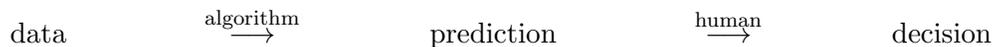

        \centering
        \begin{align*}
        &\text{data}
        &
        &\stackrel{\text{algorithm}}{\longrightarrow}
        &
        &\text{prediction}
        &
        &\stackrel{\text{human}}{\longrightarrow}
        &
        &\text{decision}
    \end{align*}
        \caption{A view on data-driven decisions}
        \label{fig:context}
    \end{figure}
    
    The need to consider properties of algorithmic decisions in their proper context often requires extending current frameworks that narrowly focus on the properties of static decision rules without modeling how they are learned from data or leveraged to make decisions.
    In this article, we instead highlight cases in which an analysis of fairness relies on a view of the whole process by which data and human beliefs are turned into decisions.
    While this framework still falls short of modeling the additional historical, institutional, and dynamic context that may be required to describe equilibrium outcomes and their welfare implications, we hope to highlight the value of analyzing critical properties of algorithmic decisions in their context.

    \section{Extensions}
    \label{sec:extensions}
    
    Our current model remains limited to prediction decisions by a rational decision-maker with biased beliefs, where predictions are provided by simple averages from unbiased training data.
    In this section, we briefly discuss extensions to our approach that aim to make the analysis more complete and applicable.

    \subsection{Correlated covariates}

    Our current model considers the inclusion or exclusion of protected group characteristics.
    However, in many applications, the inclusion of protected class itself may be prohibited or otherwise infeasible.
    In such cases, discussions around inclusions and exclusion may instead center around the role of features that are correlated with protected class, and thus may serve as proxies that lead to disparity in machine predictions.
    In those cases, similar results apply as for the role of protected class itself: excluding close proxies may hurt, rather than help, to reduce disparities when predictions become less informative about true differences across groups.
    
    \subsection{Binary decisions}
        
    Many of the assisted decisions studied in the fairness literature, such as employment, pretrial release, and lending, are binary rather than continuous. To convert between real-valued parameters and binary (or discrete) decisions, thresholds are commonly used. In these cases, we can distinguish between an algorithm's prediction, for example, a defendant's risk score for recidivism, and the decision itself, such as whether to release a defendant on bail. Even when a prediction is used automatically in decision-making, this implicitly assumes the use of a decision rule that translates the prediction to a decision, such as a risk threshold below which defendants are released. 
    We have chosen prediction decisions with a continuous true label for convenience.
    While we do not expect the main takeaways to change when applied to binary decisions based on threshold rules, we see a more complete treatment that formally includes binary decisions as a natural next step in the analysis.
    
    \subsection{Biased preferences}

    Throughout this work, we have assumed that both the decision-maker and algorithm designer are aligned in the goal of minimizing risk. As \cite{kamenica_bayesian_2011} shows, if there are misaligned preferences, the algorithm designer has the potential to improve the decisions taken by designing the structure of the information revealed to the decision-maker. An interesting extension is the optimal design of algorithms for cases where the source of bias is preference misalignment, rather than belief bias.

    \subsection{Biased data, biased equilibria}
    
    In our model, we have assumed that the training data provides unbiased signals of some ground-truth label of interest.
    But when the data itself is biased then the exclusion of protected characteristics may correct for biases and lead to more accurate predictions.
    Relatedly, the data may itself be the result of decisions by algorithms and (biased) decision-makers, leading to biased conclusions that may sustain discriminatory equilibria. Modeling biases in the data itself is therefore an important extension of our current model.
    
    \subsection{Machine learning beyond simple averages, and transparency}
    
    Our model utilizes simple averages to model machine-learning algorithms, which are easy to interpret statistically and can provide a clear answer to the problem of information transfer we wish to solve.
    An important extension of this work is to model more complex regularization schemes.
    Information between types $(x',g')$ and $(x,g)$ could be averaged if they draw from $Y$ from similar distributions. Such data combination could happen adaptively. 
    
    Modeling such an extension faces at least two challenges that we side-step in our current results.
    First, pointwise predictions could then contain information about other parts of the distribution, including direct information about differences between groups.
    Second, more complex algorithms pose the question of transparency and whether the decision-maker in our model would be able to understand the algorithm well enough to update optimally.
    Here, both transparent and intransparent algorithms could have advantages.
    
    \subsection{Noisy decisions}
    
    A fundamental problem when modeling human decision-making is the inconsistency in actions taken under seemingly similar conditions. For example, \cite{kleinberg_human_2018} find that (less noisy) predictions of (noisy) human decisions may outperform the latter in judicial decisions.
    While our current model assumes that decisions given data and features are not stochastic, our model extends naturally to noisy decisions.
    
     \subsection{Alternative notions of fairness, disparities in accuracy}
      
     We have focused so far on conditional statistical parity in predictions as our measure of fairness.
     However, disparities are only one way of expressing unfair decisions; for example, we may also be interested in differences in the accuracy of decisions across groups.
     Discussing fairness measures based on mistakes becomes particularly relevant as we move from prediction to binary decisions.
     For example, \cite{Morgan2019-td} consider equalized odds as a notion of fair treatment by the machine in their model of computer-assisted decision-making.
     The fact that different notions of fairness cannot be achieved simultaneously \citep{kleinberg_inherent_2016,chouldechova_fair_2016} suggests that the conclusions we would draw in our model could depend on the specific measure we consider.
     
    \subsection{Distributional assumptions}
    
    In our model from \autoref{sec:setup}, we currently assume that error terms are Normally distributed with fixed variance, which ensures that the error terms themselves are not informative about the means.
    (The prior belief, on the other hand, does not generally have to be Normal.)
    We conjecture that the assumptions of Normality can be replaced by a more general class of smooth distributions.
    Likewise, we believe that the restriction that the main results hold $\pi$-almost surely can be weakened to holding for any true mean vector inside an open set that has strictly positive prior probability, under regularity assumptions.

    \section{Related literature}
    \label{sec:literature}

    We contribute to an interdisciplinary literature on algorithmic fairness that spans computer science, statistics, economics, law, and operations.
    For example, \cite{kleinberg_inherent_2016,chouldechova_fair_2016} study tensions between different fairness qualities of algorithms. Similarly, \cite{Corbett-Davies2017-bx} shows that `fairness constraints' can be seen to induce the unfair practice of holding different individuals to different standards. \cite{lakkaraju_selective_2017,jiang_identifying_2020} study bias coming from label bias in the training data. \cite{liang2021algorithmic} characterizes the Pareto frontier between fairness and accuracy in an algorithm design problem.
    \cite{Corbett-Davies2018-kr} reviews prevalent measures of fairness and discusses their shortcomings, while \cite{Kleinberg2019-if} shows that algorithm simplicity can lead to inequities.
    A growing set of questions arising in this literature concern the impact different restrictions on algorithmic inputs have on the fairness of the resulting decisions \citep[e.g.][]{barocas2016big,gillis_big_2018,yang_equal_2020,hellman_measuring_2020, fu_fair_2022,kim2022race}, whether there are legal restrictions on algorithmic inputs \citep[e.g.][]{bent_algorithmic_2019,gillis_input_2021}, and the role of the human in overseeing algorithms \citep[e.g.][]{enarsson2022approaching,huq2020right}. 
    Relative to this literature, we focus on the question of how the properties of an underlying prediction algorithm affect decisions when there is a human decision-maker in the loop.

    We also relate to work on the interaction of machine predictions and human decisions \citep[as also review by][]{lai_towards_2021,hemmer_human-ai_2021}.
    Similar to our setup, \cite{bastani2021improving} studies how machine learning can improve human decision-making by learning which recommendations lead to better human decisions in a cooking game.
    \cite{ibrahim_eliciting_2021} considers information flowing from a human to a machine, finding that a signal by the human may be more useful for machine prediction than a human's direct forecast.
    \cite{athey_allocation_2020} consider conditions under which it is optimal to allocate decision authority to a human rather than an algorithm.
    \cite{Raghu2019-xs} considers the allocation of tasks between algorithms and human decision-makers, including the role of algorithms in allocating human effort efficiently.
    \cite{dietvorst_overcoming_2018,stevenson_algorithmic_2019, Ludwig2021-gq} study frictions in the adoption of machine predictions (with the latter two specifically examining criminal risk assessment programs). \cite{boyaci_human_2022,snyder_algorithm_2022} look at how cognitive limitation impacts the way decision-makers use algorithmic assistance.
    \cite{Grgic-Hlaca2019-xc,Fogliato2021-de} runs vignette experiments to analyze the effects of algorithmic advice on judicial decisions.
    \cite{garrett_judging_2018} reviews the use of criminal risk assessments in the United States and find high variance in judge reliance on the tools.
    \cite{de-arteaga_case_2020} considers instances in which humans are able to identify and override incorrect estimates of risk.
    \cite{grimon_impact_2022} finds that human decision-makers focus more on information not considered by an algorithm in a setting where social workers predict child maltreatment, reducing related child hospitalizations as a result.
     \cite{DBLP:conf/aies/LakkarajuB20} looks at how black-box algorithms can be explained in multiple ways, with different explanations causing different uptake of the algorithms' predictions.
    Relative to these approaches, we specifically focus on the interaction of machine predictions with a \textit{biased} human decision-maker.
    
    More specifically, we relate to a literature that explicitly considers fairness properties of machine-assisted human decisions.
    Most closely related to our theoretical results, \cite{Morgan2019-td} analyzes specific notions of discrimination for a machine-aided human classification decision, and show that there may be a trade-off between avoiding discrimination in the underlying machine classification and avoiding discrimination in the human decision, where only trivial cases allow for avoiding both at the same time.
    Relative to their approach, we focus on the role of biases in human beliefs, and our analysis still applies to cases where there are no true differences between groups. On the experimental side, our work is closely related to \cite{green_disparate_2019}, which examines how algorithmic risk assessments impact the performance, evaluation capabilities, and bias exhibited by human decision-makers in an online experiment. We expand upon this study by evaluating how changes to the information used by the assistant alter the human's performance and bias. \cite{Imai2020-ci} studies the effect of providing judges with risk assessment tools in pre-trial decisions and considers the impacts on the fairness of the resulting decisions.
    Unlike that work, we focus on the effect of excluding protected characteristics, rather than the effect of providing advice.
    \cite{rambachan_economic_2020} considers the regulation of human and algorithmic decision-making, including in the case where potentially discriminatory human decision-makers employ predictive algorithms.
    In contrast, we focus on the interaction of biased beliefs with restrictions on algorithms. \cite{casoria_effect_2022} runs a laboratory hiring experiment to study the impact of (imperfect) group identity information on discrimination in the decisions made. \cite{agan_ban_2018} looks at the impact of `ban the box' programs and finds that removing voluntary disclosures of past felonies from job applications increases racial discrimination in who receives callbacks.
    Relative to these studies, we ask how similar issues apply in the context of an algorithmic assistant.
    
    The article is also related to an economics literature on fairness and discrimination, which attempts to identify sources of biased decisions and distinguishes between preference- and belief-based explanations \citep{bordalo_beliefs_2019,bohren_inaccurate_2019,bohren_dynamics_2019,coffman_role_2021,chan_discrimination_2022} as well as interactions that may occur between the two \citep{eyting_why_2022}.  Prior work provides evidence of how inaccurate beliefs can lead to biases in observed decisions. Building on this work, our article considers how information from a machine prediction interacts with inaccurate beliefs.

    Finally, we relate to work on the communication of statistical results, the statistics and optimization of learning optimal decisions, as well as information design.
    Related to our approach, \cite{andrews2021model} models the communication of statistical results to a decision-maker; in the same way, we model the machine-learning algorithm as providing input to the decision-maker's choice.
    Relative to this work, we focus specifically on the relationship of the inclusion of protected characteristics to disparities in decisions when beliefs are biased. Biased information systems have been shown to be optimal when signaling is restricted to coarse (binary over a continuous space) signals \citep{meyer_learning_1991,suen_selfperpetuation_2004}.
    A related literature on forming optimal decisions from training data also emphasizes the distinction between prediction and resulting decision \citep[e.g.][]{bertsimas2020predictive}.
    More specifically, \cite{kleinberg2018algorithmic} argues that fairness constraints should be imposed in the decision, rather than the prediction stage.
    \cite{Canetti2019-fm,Mishler2021-hg} consider whether and how fairness can be achieved through post-processing. 
    Our information structure also resembles that of \cite{kamenica_bayesian_2011}, which demonstrates how an informed sender with control of information can design a signaling scheme that influences the behavior of a decision-making receiver with misaligned preferences.
    
    \section{Conclusion}
    \label{sec:conclusion}
        
    In this article, we present a model in which the fairness implications of a machine-assisted human decision cannot be assessed purely from the mathematical properties of underlying predictions.
    Instead, we argue that the nature of decision authority, beliefs, incentives, existing bias, and equilibria all shape the fairness of complex decision systems.
    When analyzing disparities of algorithmic decisions, this view provides a research agenda towards less discriminatory decisions by designing processes that consider fairness implications in all aspects of their design.
    We believe that this agenda requires bringing together techniques, ideas, scholars, and stakeholders from across fields and across application areas.
       
    \bibliography{humachine}

    \input{proofs}

\end{document}

%% file: proofs.tex
\onecolumn
\appendix

\begin{center}
    \Large
    \textsc{Appendix}
\end{center}

\section{Solution for the example}

In this section we present the calculations necessary for all computations in the example given in \autoref{sec:example}.
Here, we assume that:
\begin{itemize}
    \item $Y|G = g \sim \N(\mu(g),\sigma^2)$ is drawn independently with $\Delta_\mu = \mu(1) - \mu(0)$ and $\bar{\mu} = \frac{\mu(1) + \mu(0)}{2}$.
    \item The decision-maker holds a prior belief $\pi$ that $\mu(g) \sim \N(\beta(g),\tau^2)$ independently across $g$, with $\delta = \beta(1) - \beta(0)$ and $\bar{\beta} = \frac{\beta(1) + \beta(0)}{2}$.
    \item $\P(G = 1) = \frac{1}{2}$ independent of everything else.
    \item The training data is balanced, $n(1) =\frac{n}{2} = n(0)$.
\end{itemize}

\subsection{A convenient reparametrization}

Writing $\bar{Y}_1, \bar{Y}_0$ for the respective averages,
we have that
\begin{align*}
    \hat{\Delta} = \frac{\bar{Y}_1 - \bar{Y}_0}{2} &\sim \N(\Delta_\mu / 2, \sigma^2 / n)
    &
    \hat{\mu}=\frac{\bar{Y}_1 + \bar{Y}_0}{2} &\sim \N(\bar{\mu}, \sigma^2 / n),
    \\
    \bar{\Delta} = \Delta_\mu / 2 &\sim \N(\delta/2,\tau^2/2),
    &
    \bar{\mu} &\sim \N(\bar{\beta},\tau^2/2),
\end{align*}
which are all independent.
Furthermore, we can write:
\begin{align*}
    \hat{f}_+(g) &= \hat{\mu} + (2g-1) \hat{\Delta}
    &
    \hat{f}_- &= \hat{\mu}
    &
    \mu(g) &= \bar{\mu} + (2g-1) \bar{\Delta}
\end{align*}

\subsection{Human decisions}

We update as:
\begin{align*}
    d_0(g) = \E_\pi[\mu(g)] &= \beta(g) = \bar{\beta} + \frac{2g-1}{2} \delta \\
    \hat{d}_+(g) = \E_\pi[\mu(g) | \hat{f}_+(g)] &= \frac{n \tau^2 \bar{Y}_g + 2 \sigma^2 \beta(g)}{n \tau^2 + 2 \sigma^2} 
    =
    \frac{n \tau^2 \hat{\mu} + 2 \sigma^2  \bar{\beta}}{n \tau^2 + 2 \sigma^2} 
    +
    \frac{2g-1}{2}
    \frac{n \tau^2 2 \hat{\Delta} +  2 \sigma^2 \delta}{n \tau^2 + 2 \sigma^2} 
    \\
    \hat{d}_-(g) = \E_\pi[\mu(g) | \hat{f}_-] &= \E_\pi[\bar{\mu}|\hat{\mu}] + (2g-1) \E_\pi[\bar{\Delta} | \hat{\mu}]
    = \frac{n \tau^2 \hat{\mu} + 2 \sigma^2 \bar{\beta}}{n \tau^2 + 2 \sigma^2} + \frac{2g-1}{2} \delta
\end{align*}

\subsection{Biases and variances}

By independence, we can calculate biases and variances (given the true values of $\bar{\mu},\bar{\Delta}$) separately by average and disparity to find:
\begin{align*}
    \E[d_0(g)] - \mu(g) &= (\bar{\beta} - \bar{\mu}) + \frac{2g-1}{2} (\delta - 2 \bar{\Delta})
    &
    \Var(d_0(g)) &= 0 \\
    \E[\hat{d}_+(g)] - \mu(g) &= \frac{2 \sigma^2  (\bar{\beta} - \bar{\mu})}{n \tau^2 + 2 \sigma^2} 
    +
    \frac{2g-1}{2}
    \frac{2 \sigma^2 (\delta - 2 \bar{\Delta})}{n \tau^2 + 2 \sigma^2}
    &
    \Var(\hat{d}_+(g)) &= \frac{2 n \tau^4 \sigma^2}{(n\tau^2 + 2 \sigma^2)^2}
    \\
    \E[\hat{d}_-(g)] - \mu(g)
    &= \frac{ 2 \sigma^2 (\bar{\beta}-\bar{\mu})}{n \tau^2 + 2 \sigma^2} + \frac{2g-1}{2} (\delta - 2\bar{\Delta})
    &
    \Var(\hat{d}_-(g))
    &=
    \frac{n \tau^4 \sigma^2}{(n\tau^2 + 2 \sigma^2)^2}
\end{align*}
The results in \autoref{tab:example} follow, noting that $\E\left[\frac{2G-1}{2}\right]=0, \Var\left(\frac{2G-1}{2}\right)=\frac{1}{4}$:

\begin{center}
        \begin{tabular}{rcc}
        \toprule
            $\hat{d}$ & $\E[\Delta_{\hat{d}}]$ & $\E[r_{\hat{d}}] = \E[\E^2[\hat{d}(G) - \mu(G)|G] + \Var(\hat{d}(G)|G)] + \sigma^2$ \\
            \midrule
            $\hat{f}_-$ & $0$ & $\frac{\Delta_\mu^2}{4}+\sigma^2 \left(1+\frac{1}{n}\right)$\\
            \greyrule
            $\hat{f}_+$ & $\Delta_\mu$ & $\sigma^2 \left(1+\frac{2}{n}\right)$\\
             \midrule
            $d_0$ & $\delta$ & $(\bar{\mu} {-} \bar{\beta})^2 + \frac{(\Delta_\mu {-} \delta)^2}{4} + \sigma^2$ \\
            \greyrule
            $\hat{d}_-$ & $\delta$ & $\frac{\sigma^4 (\bar{\mu} {-} \bar{\beta})^2}{(\sigma^2 + \frac{n}{2} \tau^2)^2} + \frac{(\Delta_\mu {-} \delta)^2}{4} 
             + \sigma^2 \left(1 + \frac{n \tau^4}{4(\sigma^2 + \frac {n}{2} \tau^2)^2}\right)$
            \\
            \greyrule
            $\hat{d}_+$ & $\frac{\sigma^2 \delta + \frac{n}{2}\tau^2\Delta_\mu}{\sigma^2  + \frac{n}{2}\tau^2 }$ & $\frac{\sigma^4 \left((\bar{\mu} {-} \bar{\beta})^2 + \frac{(\Delta_\mu {-} \delta)^2}{4} \right)}{\left(\sigma^2 + \frac{n}{2} \tau^2\right)^2}  + \sigma^2 \left(1 + \frac{n \tau^4}{2(\sigma^2 + \frac{n}{2} \tau^2)^2}\right)$ \\
             \bottomrule
        \end{tabular}
\end{center}

\subsection{Proofs of the remarks in the example section}

\remexamplereversal*

\begin{proof}[Proof of \autoref{rem:examplereversal}]
    Immediate from the explicit calculations, where we note that $\E[\Delta_{\hat{d}_+}]$ is a convex combination of $\delta$ and $\Delta_\mu$.
\end{proof}

\remexampletradeoff*

\begin{proof}[Proof of \autoref{rem:exampletradeoff}]
    The results for automation are immediate from the explicit expressions above.
    For assistance we note that risks are equal for
    \begin{align*}
    &
    &
        \left(1 - \frac{\sigma^4}{\left(\sigma^2+\frac{n}{2}\tau^2\right)^2}\right) \frac{(\Delta_\mu - \delta)^2}{4}
        &=
        \frac{n \sigma^2 \tau^4}{4 \left(\sigma^2+\frac{n}{2}\tau^2\right)^2}
        \\
        &\Longleftrightarrow
        &
        \left(n \tau^2 \sigma^2 + n^2 \tau^4 /4 \right) (\Delta_\mu - \delta)^2
        &=
        n \sigma^2 \tau^4
                \\
        &\Longleftrightarrow
        &
        (\Delta_\mu - \delta)^2
        &=
        \frac{4 \sigma^2 \tau^2}{4 \sigma^2 + n \tau^2 },
    \end{align*}
    yielding the threshold for $\delta$.
\end{proof}

\section{Proofs of main results}

\remtradeoffvsdominance*
\begin{proof}
    The statements about disparities are immediate from positive variance and $\Delta_{\hat{f}_-} = 0$.
    Since we assume a constant variance $\sigma^2$, the statement about risks holds by
    \begin{align*}
        \E[r_{\hat{f}_+}(x)] &= \E\left[ \frac{\sigma^2}{n(x,G)}\right] + \sigma^2
        \\
        \E[r_{\hat{f}_-}(x)] &= 
        \E\left[\left(\frac{n(x,1-G)}{n(x,1)+n(x,0)} \Delta_\mu \right)^2\right] +
        \frac{\sigma^2}{n(x,1)+n(x,0)}  + \sigma^2,
    \end{align*}
    where we note that $\E[r_{\hat{f}_+}(x)] - \E[r_{\hat{f}_-}(x)]$ is symmetrical in $\Delta_\mu$ around zero and strictly decreasing in $|\Delta_\mu|$, positive at $\Delta_\mu =0$, and negative as $|\Delta_\mu| \rightarrow \infty$. Hence there exists such $\xi$.
\end{proof}

\thmdisreversal*

\begin{proof}
    By $\delta$-disparate beliefs and for $\bar{Y}_g = \frac{\sum_{X_i=x,G_i=g} Y_i}{\sum_{X_i=x,G_i=g} 1}$, $\bar{Y} = \frac{\sum_{X_i=x} Y_i}{\sum_{X_i=x} 1}$, $\bar{\varepsilon} = \frac{\sum_{X_i=x} \varepsilon_i}{\sum_{X_i=x} 1}$ (where $\varepsilon_i = Y_i - \mu(x,g)$),
    \begin{align*}
        \Delta_{\hat{d}_-}
        &= \E_\pi[\mu(1) | \bar{Y}] - \E_\pi[\mu(0) | \bar{Y}]
        = \E_\pi\left[\mu(1)-\mu(0) \middle| \bar{Y} \right]
        \\
        &= \E_\pi\left[\mu(1)-\mu(0) \middle| \bar{\mu} + \bar{\varepsilon} \right]
        =
        \E_\pi\left[\E_\pi\left[\mu(1)-\mu(0) \middle| \bar{\mu} , \bar{\varepsilon} \right] \middle| \bar{\mu} + \bar{\varepsilon} \right]
        \\
        &=\E_\pi[\underbrace{\E_\pi\left[\mu(1)-\mu(0) \middle| \bar{\mu} \right]}_{\geq \delta} | \bar{\mu} + \bar{\varepsilon} ]
        \geq \delta
    \end{align*}
    almost surely,
    where we have used that the $\varepsilon_i$ do not vary with $G_i$ and are not informative about the $\mu(x,g)$.
    Similarly, $\Delta_{\hat{d}_-} = \E_\pi\left[\mu(1)-\mu(0)\right] = \E_\pi\left[\E_\pi\left[\mu(1)-\mu(0) \middle| \bar{\mu} \right] \right] \geq \delta$.
    At the same time,
    \begin{align*}
        \hat{d}_+(x,g) = \E_\pi[\mu(x,g)|\bar{Y}_g]
        =
        \E_\pi[\mu(x,g)|\{Y_i;G_i=g,X_i=x\}]
        \stackrel{p}{\rightarrow}
        \mu(x,g)
    \end{align*}
    as $n(x,g) \rightarrow \infty$
    $\pi$-almost surely by Doob's posterior consistency theorem \cite[e.g.][Theorem 2.2]{Miller2018-xt}
    where we have used that $\bar{Y}_g$ is a sufficient statistic for $\{Y_i;G_i=g,X_i=x\}$.
    In particular,
    \begin{align*}
        \Delta_{\hat{d}_+}(x) = \hat{d}_+(x,1) - \hat{d}_+(x,0) \stackrel{p}{\rightarrow} \Delta_\mu(x) < \delta
    \end{align*}
    as $\min(n(x,1),n(x,0)) \rightarrow \infty$.
    The result follows.
\end{proof}

\cordispreord*

\begin{proof}[Proof of \autoref{cor:dispreord}]
    The result follows as in the two previous proofs and
    \begin{align*}
        \Delta_{\hat{f}_+}(x,g) = \hat{f}_+(x,1) - \hat{f}_+(x,0)& \stackrel{p}{\rightarrow} \Delta_\mu(x),
        &
        \Delta_{\hat{f}_-}(x,g) &\equiv 0
        .\qedhere
    \end{align*}
\end{proof}

\thmtradeoffrev*

\begin{proof}[Proof of \autoref{thm:trade_reversal}]
    Since $\pi$-almost surely
    \begin{align*}
        \hat{d}_+(x,g) &\stackrel{p}{\rightarrow} \mu(x,g)
        &
        \hat{d}_-(x,1) - \hat{d}_-(x,0)
        \geq \delta
    \end{align*}
    as $n(x,g) \rightarrow \infty$,
    by continuous mapping we have that
    \begin{align*}
        r_{ \hat{d}_+}(x) &\stackrel{p}{\rightarrow} \sigma^2
        &
        r_{ \hat{d}_-}(x)
        & \geq \P(G=1)\P(G=0) (\Delta_\mu - \delta)^2 + \sigma^2 > \sigma^2.
    \end{align*}
    For automation, for all $\epsilon > 0$
    \begin{align*}
        \hat{f}_+(x,g) &\stackrel{p}{\rightarrow} \mu(x,g),
        &
        \P(|\hat{f}_-(x,g) - \mu(x,g)| \geq \zeta \Delta_\mu -\epsilon) \rightarrow 1,
    \end{align*}
    so
        \begin{align*}
        r^0_{ \hat{f}_+}(x,g) &\stackrel{p}{\rightarrow} \sigma^2,
        &
        r^0_{ \hat{f}_-}(x,g)
        & \geq \zeta^2 (\Delta_\mu - \epsilon)^2 + \sigma^2 > \sigma^2
    \end{align*}
    for $\epsilon$ sufficiently small.
\end{proof}